\documentclass[12pt]{article}
\usepackage[round]{natbib}
\usepackage{amsmath,amsthm,verbatim,amssymb}
\usepackage{bm, times,  mathrsfs, dsfont, graphics, multirow, booktabs, url}
\usepackage{epsfig,graphicx}
\usepackage{color}
\usepackage{algorithm, algorithmic}
\usepackage{appendix}
\usepackage{hyperref}
\hypersetup{hypertex=true,
            colorlinks=true,
            linkcolor=blue,
            anchorcolor=blue,
            citecolor=blue}
\usepackage{enumerate}

\usepackage[small,compact]{titlesec}

\renewcommand\baselinestretch{1.75}

\newtheorem{theorem}{Theorem}

\newtheorem{lemma}{Lemma}

\newtheorem{remark}{Remark}

\renewcommand{\hat}{\widehat}

\newcommand{\bX}{{\bf X}}
\newcommand{\bx}{\mbox{\bf x}}

\newcommand{\bh}{\mbox{\bf h}}

\newcommand{\bb}{\Delta}

\newcommand{\bd}{\mbox{\bf d}}

\newcommand{\bY}{\mbox{\bf Y}}

\newcommand{\bbeta}{\boldsymbol{\beta}}

\newcommand{\boldeta}{\boldsymbol{\eta}}

\newcommand{\bepsilon}{\boldsymbol{\epsilon}}

\title{$\ell_0$-Regularized  High-dimensional Accelerated Failure Time 
Model}
\author{Xingdong Feng, Jian Huang, Yuling Jiao, and Shuang Zhang \footnote{Xingdong Feng is
Professor and Shuang Zhang is PhD candidate, School of Statistics and Management, Shanghai University of Finance and Economics, Shanghai  200433, China (email: feng.xingdong@mail.shufe.edu.cn). Jian Huang is Professor, Department of Statistics and Actuarial Science, University of Iowa, Iowa City, Iowa 52246, USA (email: jian-huang@uiowa.edu). Yuling Jiao is Associate Professor, School of Mathematics and Statistics, Wuhan University, Wuhan 430000, China (email: yulingjiaomath@whu.edu.cn).  This work  was supported by  NSF grant DMS-1916199,  National Natural Science Foundation of China (No. 11971292, No. 11690012 and No.11871474), and Program for Innovative Research Team of SUFE.}\\
}
\date{\today}

\begin{document}
\maketitle

\begin{abstract}
\renewcommand\baselinestretch{1.4}
{
We develop a constructive approach for $\ell_0$-penalized estimation in the sparse accelerated failure time (AFT)
model  with high-dimensional covariates.
Our proposed method is based on Stute's weighted least squares criterion combined with $\ell_0$-penalization.
This method is a computational algorithm that generates a sequence of solutions iteratively, based on  active sets derived from primal and dual information and root finding according to the KKT conditions.
We refer to the proposed method as AFT-SDAR (for support detection and root finding).
An important aspect of our theoretical results is that we directly concern the sequence of solutions generated based on the AFT-SDAR algorithm. We prove that the  estimation errors of the solution sequence decay exponentially to the optimal error bound
with high probability, as long as the covariate matrix satisfies a mild regularity condition which is necessary and sufficient for model identification  even in the setting of  high-dimensional linear regression.
We also proposed an adaptive version of AFT-SDAR, or AFT-ASDAR, which determines the support size of the estimated
coefficient in a data-driven fashion.
We conduct simulation studies to demonstrate the superior performance of the proposed method over the lasso and MCP in terms of accuracy and speed.
We also apply  the proposed method to a real data set to illustrate its application.

\bigskip

}

\end{abstract}

\section{Introduction}
\label{intro}
In survival analysis, an attractive alternative to the widely used  proportional hazards model
\citep{cox1972regression}
 is the  accelerated failure time (AFT) model
\citep{koul1981regression,wei1992accelerated, kalbfleisch2011statistical}.
The AFT model is a linear regression model in which  the response variable is usually the logarithm or a known monotone transformation of the failure time.
Let $T_{i}$ be the failure time and $\bx_i$ be a $p$-dimensional covariate vector for the $i$th subject in a random sample of size $n$. The AFT model assumes
\begin{equation*}
\ln(T_i)=\bx^{T}_{i}\bbeta^*+\epsilon_i,~~i=1,\ldots,n,
\end{equation*}
where $\bbeta^*\in\mathbb{R}^p$ is the underlying regression coefficient vector,  $\epsilon_i$'s are random error terms.
When $T_i$ is subject to right censoring, we only observe $(Y_i, \delta_i, \bx_i)$, where $Y_i=\min\{\ln(T_i), \ln(C_i)\}$, $C_i$ is the censoring time, and $\delta_i=1_{\{T_i\leq C_i\}}$ is the censoring indicator.
Assume that a random sample of i.i.d. observations ($Y_i$, $\delta_i$, $\bx_i$), $i=1,\ldots,n$,  is
available.
To estimate  $\bbeta^*$ when the distribution of the error terms is unspecified, several approaches have been proposed in the literature.
One approach is the Buckley-James estimator
\citep{buckley1979linear},
which adjusts for censored observations using the Kaplan-Meier
estimator. The second approach is the rank-based estimator
\citep{ying1993large},
which is motivated by the score function of the partial likelihood.
Another interesting alternative is the weighted least squares approach
\citep{stute1993strong,stute1996distributional},
which involves the minimization of a weighted least squares objective function.

In this paper, we focus on the high-dimensional  AFT model, where the dimension of the covariate vector can exceed the sample size.
In the high-dimensional AFT model,
many researchers have proposed methods for parameter estimation and variable selection.
For example,
\cite{huang2006regularized}
considered the LASSO
\citep{tibshirani1996regression}
 in the AFT model, based on the  weighted least squares criterion;
\cite{johnson2008variable} and \cite{johnson2008penalized}
applied the SCAD
\citep{fan2001variable}
penalty to  the rank-based estimator and Buckley-James estimator;
\cite{cai2009regularized}
proposed the rank-based adaptive LASSO
\citep{zou2006adaptive}
 method;
\cite{huang2010variable}
used the bridge penalization for the regularized estimation and variable selection;
\cite{hu2013adjusted}
 extended the MCP
\citep{zhang2010nearly}
penalty to the weighted least square estimation;
{
\cite{khan2016variable}
used  the adaptive and  weighted elastic net methods
\citep{zou2009adaptive,hong2010weighted}
based on the weighted least squares criterion.}

We propose an $\ell_0$-penalized method for estimation and variable selection in the high-dimensional AFT model.  We extend the support detection and root finding (SDAR) algorithm
 \citep{huang2018constructive}
for linear regression model to the
AFT model. For convenience, we refer to the proposed method as AFT-SDAR.
In the same spirit as the SDAR method, AFT-SDAR  is  a constructive approach to estimating the sparse and high-dimensional AFT model. This approach is a computational algorithm motivated from the KKT conditions for the $\ell_0$-penalized weighted least squares solution, and generates a sequence of solutions iteratively, based on support detection using primal and dual information and root finding.
Theoretically, we show  that the $\ell_{\infty}$-norm of  the estimation errors of the solution sequence decay exponentially to the optima order $\mathcal{O}(\sqrt{\frac{\log p}{n}})$ with high probability,  as long as the covariate matrix satisfies the weakest  regularity condition that  is necessary and sufficient for  model identification. Moreover, the estimated  support coincides with the true support  of the underlying vector regression coefficients  if the minimum absolute value of the nonzero entries of the target  is above the  detectable order.

The rest of this paper is organized as follows. In Section \ref{AFT-L0},  we described the $\ell_0$-penalized criterion for the AFT model.
In Section \ref{sec:1},  we give the KKT conditions for the $\ell_0$-penalized weighted least squares solutions and describe the proposed AFT-SDAR algorithm.
In Section \ref{sec:2},  we first establish the finite-step and deterministic
error bounds for
the solution sequence generated by the AFT-SDAR algorithm.
As a consequence of these deterministic error bounds, we provide nonasymptotic error bounds for the solution sequence.
We also show that the proposed method recovers the support of the underlying regression coefficient vector in finite iterations with high probability.
In Section \ref{sec:5},  we describe
 AFT-ASDAR, the adaptive version of AFT-SDAR that selects the tuning parameter in a data driven fashion. In Section \ref{sec:6},  we assess the finite sample performance of the proposed method with different simulation studies and a real case study on a breast cancer gene expression data set.
 Concluding remarks are given in Section \ref{sec:15}. Proofs for all the lemmas and theorems are deferred to Appendix. {An R package implementing the proposed method is available at \url{https://github.com/Shuang-Zhang/ASDAR/}.}

\section{AFT regression with $\ell_0$-penalization}
\label{AFT-L0}

Let $Y_{(1)},\ldots,Y_{(n)}$ be the order statistics of $Y_i$'s. Let $\delta_{(1)},\ldots,\delta_{(n)}$  be the associated censoring indicators and let $\bx_{(1)},\ldots,\bx_{(n)}$ be the associated covariates.
In the weighted least squares method,
the weights $w_{(i)}$'s  are the jumps in Kaplan-Meier estimator based on
$(Y_{(i)}, \delta_{(i)})$, $i=1, \ldots, n$, which can be expressed as
\begin{equation}\label{wq}
\begin{split}
&w_{(1)}=\frac{\delta_{(1)}}{n},\\
&w_{(i)}=\frac{\delta_{(i)}}{n-i+1}\cdot\prod_{j=1}^{i-1}{\left(\frac{n-j}{n-j+1}\right)^{\delta_{(j)}}},i=2,\ldots,n.
\end{split}
\end{equation}
The weighted least squares criterion is given by
\begin{equation*}
{\mathcal L}_1(\bbeta)=
\frac{1}{2n}\sum_{i=1}^{n}
w_{(i)}\big{(}Y_{(i)}-\bx^{T}_{(i)}\bbeta\big{)}^2.
\end{equation*}
In the low-dimensional settings with $n \gg p$, this criterion leads to a consistent and asymptotically normal estimator under appropriate conditions 
\citep{stute1993strong,stute1996distributional}.
However, in the high-dimensional settings when $p \gg n$, regularization is needed to ensure a unique solution in
minimizing $\mathcal{L}_1(\bbeta)$.

We consider the $\ell_0$-regularized method for variable selection and estimation in AFT based on the weighted least squares criterion.
The $\ell_0$-penalized estimator is given by
\begin{equation}\label{eq03}
\bbeta^{\diamond}=
\underset{\bbeta\in \mathbb{R}^p}{\mbox{min}}~
{\mathcal L}_1(\bbeta)+
\lambda\|\bbeta\|_0,
\end{equation}
where $\lambda \ge 0$ is a tuning parameter, and  $\|\bbeta\|_0$ denotes the number of nonzero elements of $\bbeta$.

To facilitate computation, we rewrite the weighted least squares loss as a standard least squares loss as follows. Let the design matrix be $\bX=\left(\bx_{(1)},\ldots,\bx_{(n)}\right)^T$ and let $\bY=\left(Y_{(1)},\ldots,Y_{(n)}\right)^T$. Define
\begin{equation*}
\begin{split}
&\tilde{\bX}=\mbox{diag}\left(\sqrt{w_{(1)}},\ldots,\sqrt{w_{(n)}}\right)\cdot\bX,\\
&\bar{\bY}=\mbox{diag}\left(\sqrt{w_{(1)}},\ldots,\sqrt{w_{(n)}}\right)\cdot \bY.
\end{split}
\end{equation*}
Without loss of generality, assume that $\|\tilde{\bx}_{j}\|_2>0$, $j=1,\ldots,p$, hold throughout this paper, where $\tilde{\bx}_{j}$ is the $j$th column of $\tilde{\bX}$. Let
 \begin{equation*}
D=\mbox{diag}\Big{(}\frac{\sqrt{n}}
{\|\tilde{\bx}_{1}\|_2},\ldots,\frac{\sqrt{n}}{\|\tilde{\bx}_{p}\|_2}\Big{)}.
\end{equation*}
Define
$\boldeta=D^{-1}\bbeta$ and $\bar{\bX}=\tilde{\bX} D$.
Then  each column of $\bar{\bX}$ is $\sqrt{n}$-length and supp($\boldeta$)=supp($\bbeta$), where supp($\bbeta$)=$\{j: \beta_j\neq 0, j=1,\ldots,p\}$.
Let
\[
\mathcal{L}_2(\boldeta)=\frac{1}{2n}\left\|\bar{\bY}-\bar{\bX}\boldeta\right\|_2^2.
\]
Define
\begin{equation}\label{eq1}
\boldeta^{\diamond}=\underset{\boldeta\in \mathbb{R}^p}{\mbox{min}}~
\mathcal{L}_2(\boldeta)
+\lambda\|\boldeta\|_0,
\end{equation}
Then the estimator of $\bbeta$ defined in \eqref{eq03} can be obtained as
$\bbeta^{\diamond}=D \boldeta^{\diamond}.$

\section{AFT-SDAR Algorithm}
\label{sec:1}
We first introduce some notation used throughout the paper.
Let $\|\boldeta\|_q=(\sum_{i=1}^{p}|\eta_{i}|^q)^{\frac{1}{q}}$ be the usual $q$ ($q\in [1,\infty]$) norm
of the vector $\boldeta=(\eta_1,\ldots,\eta_p)^{T}\in \mathbb{R}^{p}$.
Let $|A|$ denote the cardinality of the set $A$. Denote $\boldeta_A=(\eta_i,i\in A)\in \mathbb{R}^{|A|}$, $\boldeta|_A\in \mathbb{R}^{p}$ with its $i$th element $({\boldeta|_A})_i=\eta_i\text{1}(i\in A)$, where $\text{1}(\cdot)$ is the indicator function. Let
$\|\boldeta\|_{T,\infty}$ and $\|\boldeta\|_{\min}$
be the $T$th largest elements (in absolute value) and the minimum absolute value of $\boldeta$, respectively. Let $\|M\|_{\infty}$ denote the maximum value (in absolute value) of the matrix $M$. Let $\nabla  \mathcal{L}$ denote the gradient of
function $\mathcal{L}$. Denote $\bar{\bX}_A=(\bar{\bx}_j,j\in A)\in \mathbb{R}^{n\times |A|}$, where $\bar{\bx}_j$ is a column of a matrix $\bar{\bX}$. Let $\sigma_{\min}(\bar{\bX}_A^T\bar{\bX}_A)$ be the minimum eigenvalue of the matrix $\bar{\bX}_A^T\bar{\bX}_A$ and $\|\bar{\bX}_A\|_2$ be the spectrum norm of the matrix $\bar{\bX}_A$.
{Denote
$\sigma_{(\min,2T)}
=$$\min\{\sigma_{\min}(\bar{\bX}_{2T}^T\bar{\bX}_{2T}): \bar{\bX}_{2T}\in \mathbb{R}^{n\times 2T}$
$\mbox{consists}~ \mbox{of}~2T~\mbox{columns}~\mbox{of}~ \bar{\bX}\}$.}

The following lemma gives the KKT conditions of the minimizer of \eqref{eq1}.
\begin{lemma}\label{L1}
If $\boldeta^\diamond$ is a minimizer of \eqref{eq1}, then $\boldeta^\diamond$ satisfies
\begin{align}\label{eq2}
\left\{
\begin{aligned}
&\bd^{\diamond}=\bar{\bX}^{T}(\bar{\bY}-\bar{\bX}\boldeta^{\diamond})/n,\\
&\boldeta^{\diamond}=H_{\lambda}(\boldeta^{\diamond}+\bd^{\diamond}), \\
\end{aligned}
\right.
\end{align}
where the $i$th element of $H_{\lambda}(\cdot)$ is defined by
\begin{align}\label{eq3}
 (H_{\lambda}(\boldeta))_{i}=\left\{
\begin{aligned}
&0  ,&&&|\eta_{i}|<\sqrt{2\lambda},\\
&\eta_{i}  ,&&&|\eta_{i}|\geq\sqrt{2\lambda}.\\
\end{aligned}\right.
\end{align}
Conversely, if $\boldeta^{\diamond}$ and $\bd^{\diamond}$ satisfy \eqref{eq2}, then $\boldeta^{\diamond}$ is a local minimizer of \eqref{eq1}, and
$\bbeta^{\diamond}=D\boldeta^{\diamond}$ is a local minimizer of \eqref{eq03}.
\end{lemma}


Our proposed AFT-SDAR algorithm is based on solving the KKT equations
\eqref{eq2} iteratively.
Let $A^\diamond=\text{supp}(\boldeta^\diamond)$ and $I^\diamond=(A^\diamond)^c$.
Based on \eqref{eq2} and by the definition of $H_{\lambda}(\cdot)$, we have
\begin{equation*}
A^\diamond=\{i:|\eta^\diamond_{i}+ d^\diamond_{i}|\geq\sqrt{2\lambda}\},~~I^\diamond=\{i:|\eta^\diamond_{i}+ d^\diamond_{i}|<\sqrt{2\lambda}\},
\end{equation*}
and
\begin{align}
\label{eq5a}
\left\{
\begin{aligned}
&\boldeta_{I^\diamond}^{\diamond} =0\\
&\bd_{A^\diamond}^{\diamond} =0\\
&\boldeta_{A^\diamond}^{\diamond}=(\bar{\bX}_{A^\diamond}^{T}\bar{\bX}_{A^\diamond})^{-1}\bar{\bX}_{A^\diamond}^{T}\bar{\bY}\\
&\bd_{I^\diamond}^{\diamond}=\bar{\bX}_{I^\diamond}^{T}(\bar{\bY}-\bar{\bX}_{A^\diamond}\boldeta_{A^\diamond}^{\diamond})/n\\
&\bbeta^{\diamond}=D\boldeta^{\diamond}.
\end{aligned}
\right.
\end{align}
We solve these equations iteratively.
Let $\{\boldeta^{k},\bd^{k}\}$ be the values at the $k$th iteration,
and let $\{A^{k},I^{k}\}$ be the active and inactive sets based on
$\{\boldeta^{k},\bd^{k}\}$, where
\begin{equation}\label{eq4}
\begin{split}
&A^k=\{i:|\eta^k_{i}+ d^k_{i}|\geq\sqrt{2\lambda}\},\\
&I^k=\{i:|\eta^k_{i}+ d^k_{i}|<\sqrt{2\lambda}\}.
\end{split}
\end{equation}
Then based on \eqref{eq5a}, we calculate the updated values
\begin{equation*}
\{\boldeta^{k+1}_{I^{k}},\bd^{k+1}_{A^{k}},\boldeta^{k+1}_{A^{k}},\bd^{k+1}_{I^{k}},\bbeta^{k+1}\},
\end{equation*}
as follows:
\begin{align}\label{eq5}
\left\{
\begin{aligned}
&\boldeta^{k+1}_{I^{k}}=0\\
&\bd^{k+1}_{A^{k}} =0\\
&\boldeta^{k+1}_{A^{k}}=(\bar{\bX}_{A^k}^{T}\bar{\bX}_{A^k})^{-1}\bar{\bX}_{A^k}^{T}\bar{\bY}\\
&\bd^{k+1}_{I^{k}}=\bar{\bX}_{I^k}^{T}(\bar{\bY}-\bar{\bX}_{A^k}\boldeta_{A^k}^{k+1})/n\\
&\bbeta^{k+1}=D\boldeta^{k+1}.\\
\end{aligned}
\right.
\end{align}
Suppose that $\|\boldeta^*\|_{0}=\|\bbeta^*\|_{0}=K\leq T$ for some $T \ge 1$,
where $\boldeta^*=D^{-1}\bbeta^*$.
At the $k$th iteration, we set
\begin{align}\label{eq6}
\sqrt{2\lambda}=\parallel\boldeta^k+ \bd^k\parallel_{T,\infty}
\end{align}
 in \eqref{eq4}. Hence $|A^k|=T$ in every iteration due to this $\lambda$. Note that the tuning parameter $\lambda$ is expressed in terms of $T$. We will use a data-driven procedure to tune the cardinality $T$ in Section \ref{sec:5}.

Let $\boldeta^0=D^{-1}\bbeta^0$ be an initial value, then we get a sequence of solutions $\{\boldeta^k,k\geq1\}$ by using \eqref{eq4}
 and \eqref{eq5} with the value of $\lambda$ given in \eqref{eq6}.
We introduce a step size $0 < \tau \le 1$
 in the definitions of the active and inactive sets as follows:
 \begin{equation}
 \label{eq4b}
\begin{split}
&A^k=\{i:|\eta^k_{i}+ \tau d^k_{i}|\geq\sqrt{2\lambda}\},\\
&I^k=\{i:|\eta^k_{i}+ \tau d^k_{i}|<\sqrt{2\lambda}\},
\end{split}
\end{equation}
with $\sqrt{2\lambda}=\parallel\boldeta^k+ \tau \bd^k\parallel_{T,\infty}$.
The step size $\tau$ plays the role of weighing the importance of  $\boldeta^k$ and $\bd^k$ in determining the active and inactive sets.

{In order to bound the estimation error of the sequences generated by AFT-SDAR, we need some regularity conditions on the covariate matrix. 
 Thanks to
 this step size $\tau$, we can replace the sparse Riesz condition  used in 
\cite{huang2018constructive}
 in analyzing the SDAR   to the weakest condition possible, which  is necessary and  sufficient for model identification even in high-dimensional linear regression, see Section \ref{sec:2}  for detail.
}

We describe the AFT-SDAR algorithm in detail in  Algorithm \ref{alg:1}.
\begin{algorithm}[!ht]
	\caption{AFT-SDAR}
    \label{alg:1}
	\begin{algorithmic}[1]
\STATE Input: $\boldeta^0=D^{-1}\bbeta^0$, $\bd^0=\bar{\bX}^T(\bar{\bY}-\bar{\bX}\boldeta^0)/n$, $\tau$, T; $k=0$
		\FOR{$k= 0,1,\ldots,$}
\STATE $A^k=\big{\{}j:|\eta^k_j+\tau d^k_{j}|\geq\|\boldeta^k+\tau \bd^k\|_{T,\infty}\big{\}}$, $I^{k}=(A^{k})^c$.
\STATE $\boldeta^{k+1}_{I^k}=0$.
\STATE $\bd^{k+1}_{A^k}=0$. \STATE $\boldeta^{k+1}_{A^{k}}=(\bar{\bX}_{A^k}^{T}\bar{\bX}_{A^k})^{-1}\bar{\bX}_{A^k}^{T}\bar{\bY}$.
\STATE $\bd^{k+1}_{I^{k}}=\bar{\bX}_{I^k}^{T}(\bar{\bY}-\bar{\bX}_{A^k}\boldeta_{A^k}^{k+1})/n$.
\STATE $\bbeta^{k+1}=D\boldeta^{k+1}$.
\STATE $\textbf{if}~A^{k}=A^{k+1}$, \textbf{then}
\STATE Stop and denote the last iteration $\hat{\bbeta}$, $\boldeta_{\hat{A}}$, $\boldeta_{\hat{I}}$, $\bd_{\hat{A}}$, $\bd_{\hat{I}}$.
\STATE \textbf{else}
\STATE $k=k+1$
\STATE \textbf{end if}
\ENDFOR
\STATE Output:  $\hat{\bbeta}=D\cdot\big{(}\boldeta^{\mathrm{T}}_{\hat{A}}, \boldeta^{\mathrm{T}}_{\hat {I}}\big{)}^{\mathrm{T}}$ and $\hat{\boldeta}=\big{(}\boldeta^{\mathrm{T}}_{\hat {A}}, \boldeta^{\mathrm{T}}_{\hat {I}}\big{)}^{\mathrm{T}}$~as the estimates of $\bbeta^*$ and $\boldeta^*$ respectively.
\end{algorithmic}
\end{algorithm}

In Algorithm \ref{alg:1}, we terminate the computation when $A^{k}=A^{k+1}$ for some $k$, because the solution sequence generated by AFT-SDAR will not change afterwards. In Section \ref{sec:2}, we provide sufficient conditions under which
$A^{k}=A^{k+1}=A^*$ with high probability, where $A^*=\text{supp}(\boldeta^*)=\text{supp}(\bbeta^*)$,
that is, the support of the underlying regression coefficient can be recovered in finite many steps.
\section{Theoretical Properties}
\label{sec:2}
In this section, we consider the finite-step error bound for the solution sequence computed based on Algorithm \ref{alg:1}.
We also study the probabilistic and nonasymptotic $\ell_{\infty}$  error bound
for the solution sequence.

We first consider the deterministic error bounds for the solution sequence generated
based on AFT-SDAR.
{ We choose  the step size $\tau$ satisfies}
 \begin{equation}
 \label{stepa}
 0 < \tau<\frac{1}{ \sqrt{T}U}
 \end{equation}
with $ U \ge {\|\bar{\bX}\|_2^2}/{n}$, and let $L$ be a constant satisfying
\begin{equation}
\label{stepb}
0<L\leq\frac{\sigma_{(\min,2T)}}{n\sqrt{2T}}.
\end{equation}
\begin{theorem}\label{th1}
Suppose $T \ge K$ and set $\bbeta^0=0$ in Algorithm \ref{alg:1}.
Suppose \eqref{stepa} and \eqref{stepb} hold.
For the solution at the $k$th iteration in Algorithm \ref{alg:1}, 
we have
{
\begin{align}
\label{etaerror}
\|\boldeta^k-\boldeta^*\|_{\infty}\leq& \sqrt{(K+T)(1+U/L)}(\sqrt{\xi})^k\|\boldeta^{*}\|_{\infty}+\frac{2}{L}\|\nabla \mathcal{L}_2(\boldeta^{*})\|_{\infty},
\end{align}
\begin{align}
\label{betaerror}
\|\bbeta^k-\bbeta^*\|_{\infty}\leq& \|D\|_{\infty}^2\sqrt{(K+T)(1+U/L)}(\sqrt{\xi})^k\|\bbeta^{*}\|_{\infty}
+\frac{2\|D\|_{\infty}^2}{L}\|\nabla {\mathcal L}_1(\bbeta^{*})\|_{\infty},
\end{align}
where $\xi=1-\frac{2\tau L(1-\tau \sqrt{T}U)}{ \sqrt{T}(1+K)}\in(0,1).$
}
\end{theorem}
We observe that the error bound consists of two terms as indicated in Theorem \ref{th1}. For any given values
of observations, the first term converges to
zero exponentially. The magnitude of the second term is determined by
$\nabla {\mathcal L}_2(\boldeta^{*})$ in
\eqref{etaerror}
and
$\nabla {\mathcal L}_1(\bbeta^{*})$ in \eqref{betaerror},
which are given by the gradient of the weighted least squares criterion at the underlying parameter value. Therefore, under the model assumption, their expected values are zero and should be concentrated in a small neighborhood of zero.

To study the probabilistic and nonasymptotic error bounds of the solution sequences $\boldeta^k$ and $\bbeta^k$, we make the following assumptions.
\begin{enumerate}[(C1)]
\item
\label{cond2}
There exists a constant $b \in (0,\infty)$ such that $\|D\|_{\infty}\leq b$.
\item
\label{cond3}
The error terms $\epsilon_1,\ldots,\epsilon_n$ are independent and   identically distributed with mean zero and finite variance $\sigma^2$. Furthermore, they are subgaussian, in the sense that there exist some constants $K_1, ~K_2>0$ such that $P(|\epsilon_i|>C)\leq K_2\exp(-K_1C^2)$ for all $C\geq 0$ and all $i$.
\item
\label{cond4}
The covariates are bounded, that is, there exists a constant $B>0$ such that $\max_{1\le i \le n, 1\le j \le p} |x_{ij}|\leq B.$
\item
\label{cond5}
The  error terms $(\epsilon_1,\ldots,\epsilon_n)$ are independent of the Kaplan-Meier weights $(w_1,\ldots,w_n)$.
\item
\label{cond7}
There exists some positive constants $C_1$ and $C_2$ such that
$\|\bbeta^*_{A^*}\|_{\min}\geq \frac{3b^2}{L}\varepsilon_1$, where
\begin{equation*}
\varepsilon_1= C_1\left(\frac{\log(p)}{n}\right)^{\frac{1}{2}}\left(\sqrt{\frac{2C_2\log(p)}{n}}+\frac{4\log(2p)}{n}+C_2\right)^{\frac{1}{2}}.
\end{equation*}
\end{enumerate}
\begin{remark}
Condition
(C\ref{cond2})
 constrains the maximum absolute value of the matrix $D$.
 Condition (C\ref{cond3}) on the subgaussion tails of the error terms is standard in high-dimensional regression models. Condition (C\ref{cond4}) is assumed for technical convenience. It can be relaxed to $\max_{1\le i \le n, 1\le j \le p} |x_{ij}|\leq B$ with high probability.
 Moreover, conditions 
 (C\ref{cond3})-(C\ref{cond5}) are assumed to ensure that
 $\|\nabla {\cal L}_1(\bbeta^*)\|_\infty$ is small.
 Condition
 (C\ref{cond7}) assumes that the signal is not too small, which
 is  needed for the target signal to be detectable.
\end{remark}
\begin{theorem}\label{th1b}
Suppose $T \ge K$ and set $\bbeta^0=0$ in Algorithm \ref{alg:1}.
Suppose \eqref{stepa} and \eqref{stepb} hold.
If
(C\ref{cond2})-(C\ref{cond7})
hold,
then with probability at least
$1- [{(\log p )}/{n}]^{{1}/{4}}$,
\begin{align*}
\|\bbeta^k-\bbeta^*\|_{\infty}\leq & b^2\sqrt{(K+T)(1+U/L)}(\sqrt{\xi})^k\|\bbeta^{*}\|_{\infty}+\frac{2b^2}{L}\varepsilon_1.
\end{align*}
Therefore,
$$\|\bbeta^k-\bbeta^*\|_{\infty}\leq \mathcal{O}\left( \sqrt{\frac{\log(p)}{n}}\right)$$
with high probability if $k \geq \mathcal{O}\left( \log_{\frac{1}{\xi}}\frac{n}{\log(p)} \right).$
\end{theorem}

{
To derive   the sharp estimation error bound in Theorem \ref{th1b}, we need $\xi \in (0,1)$. This is guaranteed
by choosing $\tau $ stratifying  \eqref{stepa} and $L$ satisfying  \eqref{stepb}, which only requires $\sigma_{(\min,2T)} >0$. This is a weakest possible condition even  in high-dimensional linear regression model $\bar{\bY} = \bar{\bX}\bbeta^* + \bepsilon$, where, $\|\bbeta^*\|_0 \leq T$, since $\sigma_{(\min,2T)} >0$  is equivalent to the condition that the linear model is identifiable.
To be precise, let $\widetilde{\bY} = \bar{\bX}\widetilde{\bbeta}^* +\bepsilon$  and   $\|\widetilde{\bbeta}^*\|_0 \leq T.$   If we wish to derive $\bbeta^* = \widetilde{\bbeta}^*$ from $\widetilde{\bY }= \bar{\bY}$, i.e.,  from $\bar{\bX}(\bbeta^* - \widetilde{\bbeta}^*)=0$, we need $\sigma_{(\min,2T)} >0$, which is
a sufficient and necessary condition.
However, in the analysis of  SDAR algorithm
\citep{huang2018constructive},
the authors assumed  stronger conditions, i.e., sparse Riesz condition (SRC)
\citep{ZhangCun-Hui2008Tsab}
to obtain the  estimation error bound. The condition $\sigma_{(\min,2T)} >0$ is also  weaker  than the  kinds of restricted strong convexity conditions used in
bounding the estimation error for the global solutions in penalized convex and nonconvex regressions, see  \cite{Zhangzhang}, \cite{mw2019} and the references therein.
}

This result gives nonasymptotic error bound of the solution sequence.
In particular, when $\log p = o(n)$, the solution sequence converges
to the underlying regression coefficient with high probability.

The following theorem establishes the support recovery property of
AFT-SDAR.
\begin{theorem}
\label{th2}
Suppose $T \ge K$ and set $\bbeta^0=0$ in Algorithm \ref{alg:1}.
Suppose \eqref{stepa} and \eqref{stepb} hold.
If 
(C\ref{cond2})-(C\ref{cond7})
hold,
then with probability at least
$1- [{(\log p )}/{n}]^{{1}/{4}}$,
$A^*\subseteq A^k$ if $k> \log_{\frac{1}{\xi}} 9 (T+K)(1+U/L)r^2 b^4$, where $r = \frac{\|\bbeta^*\|_{\infty} }{\|\bbeta^*_{A^*}\|_{\min}}$ is the ratio of the largest absolute coefficient over the smallest absolute nonzero coefficient of $\bbeta^*$.
\end{theorem}

Theorem \ref{th2}  demonstrates that the estimated support via AFT-SDAR will contain the true support with the cost at most   $\mathcal{O}(\log (T))$ number of iterations if the minimum signal strength of $\bbeta^*$ is above the detectable threshold  $ \mathcal{O}(\sqrt{\frac{\log(p)}{n}})$.
Further, if we set $T = K$ in AFT-SDAR, then the stopping condition $A^k = A^{k+1}$ will hold if $k \geq  \mathcal{O}(\log (K))$ since the estimated supports coincide with the true support then.
As a consequence, the Oracle estimator will be recovered in    $\mathcal{O}(\log (K))$ steps.

Finally, we note that an important aspect of the results above is that they directly concern the sequence of solutions generated based on Algorithm \ref{alg:1}, rather than a theoretically defined global solution to the nonconvex $\ell_0$-penalized weighted least squares criterion. Thus there is no gap between our theoretical results and computational algorithm.

\section{Adaptive AFT-SDAR}
\label{sec:5}
In practice, the sparsity level of the true parameter value $\boldeta^*$ or $\bbeta^*$ is unknown. Therefore, we can regard $T$ as a tuning parameter.
Let $T$ increase from 0 to $Q$, which is a given large enough integer. In general, we set $Q=\alpha n/\log(n)$ as suggested by
\cite{fan2008sure},
where $\alpha$ is a positive constant.
Then we can obtain a set of solutions paths: $\{\widehat{\boldeta}(T):T=0,1,\ldots,Q\}$, where $\widehat{\boldeta}(0)=0$. Finally, we use the cross-validation method or HBIC criteria
\citep{wang2013calibrating}
to determine $\widehat{T}$,  the value of $T$. Thus we can take $\widehat{\boldeta}$ with $T=\widehat{T}$ as the estimate of $\boldeta^{*}$.
We can  also run Algorithm  \ref{alg:1} until $\|\boldeta^k-\boldeta^{k + 1}\|<\varepsilon$ by increasing $T$, where $\varepsilon$ is a given tolerance level. Then $\boldeta^k$ can be taken as the estimation of $\boldeta^*$.
Furthermore, we can gradually increase $T$ to run Algorithm \ref{alg:1} until the residual sum of squares is less than a given tolerate level $\varepsilon$, then  output $\boldeta^k$ at this time to terminate the calculation. In summary, we get an adaptive AFT-SDAR algorithm as described in Algorithm \ref{alg:2}.
\begin{algorithm}[!ht]
\caption{AFT-ASDAR}
\label{alg:2}
\begin{algorithmic}[1]
\STATE Input: $\boldeta^0=D^{-1}\bbeta^0$, $\bd^0=\bar{\bX}^T(\bar{\bY}-\bar{\bX}\boldeta^0)/n$, $\tau$, an integer $\vartheta$, an integer Q, an early stopping criterion (optional). Set $k=1$.
\FOR{$k=1,2,\ldots,$}
\STATE Run Algorithm \ref{alg:1} with~$T=\vartheta k$ and with initial value $\boldeta^{k-1}$, $\bd^{k-1}$. Denote the output  by $\boldeta^k$, $\bd^k$.
\STATE \textbf{if} the early stopping criterion is satisfied or $T>Q$,
\textbf{then}
\STATE \quad stop
\STATE \textbf{else}
\STATE \quad $k=k+1$
\STATE \textbf{end if}
\ENDFOR
\STATE Output: $\widehat{\boldeta}(\widehat{T})$ and $\widehat{\bbeta}(\widehat{T})=D\cdot\widehat{\bbeta}(\widehat{T})$~~as the estimates of $\boldeta^*$ and $\bbeta^*$ respectively.
\end{algorithmic}
\end{algorithm}
\section{Numerical studies}
\label{sec:6}
In this section, we conduct simulation studies  and real data analysis
to illustrate the effectiveness of the proposed method. We compare the  simulation results of AFT-SDAR/AFT-ASDAR with those of Lasso and MCP in  terms of accuracy and efficiency.
We also evaluate the performance
in terms of the effect of the  model parameters including the sample size $n$,  the variable dimension $p$, the correlation measure $\rho$ among covariates and the censoring rate  $c.r$.
Moreover,  we examine the average number of iterations for AFT-SDAR
to converge.
We also apply AFT-SDAR to a real data set to illustrate its application.
We implemented Lasso and MCP for the AFT model
using the coordinate descent algorithm 
\citep{breheny2011coordinate}.

\subsection{Accuracy and efficiency}
\label{sec:7}
We  generate a $n\times p$ random Gaussian matrix $\widetilde{\bX}$
whose entries are i.i.d. $\sim N(0,1)$. Then the design matrix $\bX$ is generated with $\bx_1=\widetilde{\bx}_1$, $\bx_p=\widetilde{\bx}_p$, and $\bx_j=\widetilde{\bx}_j+\rho(\widetilde{\bx}_{j+1}+\widetilde{\bx}_{j-1})$, $j=2,\ldots,p-1$. Here $\rho$ is a measure of the correlation among covariates. The underlying regression coefficient vector $\bbeta^*$ with
$K$ nonzero coefficients is generated
such that the $K$ nonzero coefficients in $\bbeta^*$ are uniformly
distributed in $(m_1,m_2)$, where $m_1=\sigma\sqrt{2\log{p}/n}$
and $m_2=100\cdot m_1$.
The $K$ nonzero coefficients are  randomly assigned to the $K$ components of $\bbeta^*$.
For each subject,
the responses $\ln(T_{i})=\bx_i^T\bbeta+\epsilon_i$, where $\epsilon_i$ is generated independently from $N(0,\sigma^2)$,
and the censoring variable $C_i$ is generated independently from the uniform distribution $U(0,\eta)$, where $\eta$ controls the censoring rate such that the desired censoring rate can be obtained.
We compare AFT-SDAR, AFT-ASDAR with Lasso and MCP  on the data generated from theses models.
In the implementation of AFT-ASDAR, we set $Q=n/\log(n)$, and terminate the computation if the residual $\|\bar{\bY}-\bar{\bX}\boldeta^k\|_2$ is smaller than $\varepsilon=\sqrt{n}\sigma$.
To examine the effect of  the correlation measure $\rho$,
we set $n=500$, $p=10000$, $K=20$, $\sigma=1$, $c.r=0.3$ and $\rho=0.3:0.3:0.9$, i.e., $\rho$ takes a grid of values from 0.3 to 0.9 with a step size 0.3.
\begin{table}
\caption{Numerical results (the averaged relative error, CPU time) on data sets with $n=500$, $p=10000$, $K=20$, $\sigma=1$, $c.r=0.3$, $\rho=0.3:0.3:0.9$.}
\label{tab:1}
\begin{center}
\scalebox{0.91}{
\begin{tabular}{llrrr}
\hline\noalign{\smallskip}
 $\rho$  & Method & ReErr ($10^{-2}$) & Time(s)\\
\noalign{\smallskip}\hline\noalign{\smallskip}
0.3&Lasso& 10.49&10.57\\
&MCP&1.10&11.55\\
&AFT-SDAR ($\tau=1$)&0.51&4.44\\
&AFT-ASDAR ($\tau=1$)&0.52&4.60\\
&AFT-SDAR ($\tau=0.5$)&0.51&4.46\\
&AFT-ASDAR ($\tau=0.5$)&0.52&4.72\\
\noalign{\smallskip}\hline\noalign{\smallskip}
0.6&Lasso&11.07&12.95\\
&MCP&2.10&10.91&\\
&AFT-SDAR ($\tau=1$)&2.01&4.32\\
&AFT-ASDAR ($\tau=1$)&2.02&4.53\\
&AFT-SDAR ($\tau=0.5$)&1.93&4.62\\
&AFT-ASDAR ($\tau=0.5$)&1.93&4.89\\
\noalign{\smallskip}\hline\noalign{\smallskip}
0.9&Lasso& 11.40&10.78\\
&MCP&1.08&11.45\\
&AFT-SDAR ($\tau=1$)&0.65&4.42\\
&AFT-ASDAR ($\tau=1$)&0.65&4.63\\
&AFT-SDAR ($\tau=0.5$)&0.47&4.72\\
&AFT-ASDAR ($\tau=0.5$)&0.47&5.01\\
\noalign{\smallskip}\hline
\end{tabular}
}
\end{center}
\end{table}


Table \ref{tab:1} shows the results
based on 100 independent replications of AFT-SDAR, AFT-ASDAR,
Lasso and MCP.
In Table \ref{tab:1}, the first column  gives the values of $\rho$,  the second column depicts the methods, the third column shows the averaged relative error (ReErr=$\frac{1}{100}\sum\|\hat{\bbeta}-\bbeta^*\|/\|\bbeta^*\|$), and
the fourth column shows the averaged CPU time.

It is clear from Table \ref{tab:1} that both AFT-SDAR  and AFT-ASDAR tend to have smaller  relative errors (ReErr) than those of Lasso and MCP.
When $\rho=0.6$ and 0.9,  AFT-SDAR  and AFT-ASDAR  have smaller relative errors at $\tau=0.5$   than  at $\tau=1$.
In terms of the speed,
AFT-SDAR and AFT-ASDAR are more than  twice as fast as Lasso and MCP for each $\rho$ and $\tau$, respectively.
For a wide range of the correlation measure $\rho$ and the step size $\tau$, AFT-SDAR and AFT-ASDAR perform well
in terms of relative error and computational speed. In addition, for
data with high correlations, choosing a step size less than the default value $1$ can lead to smaller relative errors.

\subsection{Support recovery}
\label{sec:8}
We now assess the support recovery performance
of AFT-ASDAR, Lasso  and MCP.
In AFT-ASDAR, we set the largest size of the support $Q=n/\log(n)$ \textcolor{black}{and the step size $\tau=1$}, and
use the HBIC criteria  to chose the cardinality $T$.
In the data generating models,
the rows of the $n\times p$ design matrix $\bX$ are i.i.d. $N(0,\Sigma)$, where $\Sigma_{ij}=\rho^{|i-j|}$, $1\leq i,j\leq p$.
 Let $R=m_2/m_1$, where $m_2=
\|\bbeta^*_{A^*}\|_{\max}$ and $m_1=\|\bbeta^*_{A^*}\|_{\min}=1$. The underlying regression coefficient vector $\bbeta^*\in \mathbb{R}^p$ is
generated in such a way that
the $K$ nonzero coefficients in $\bbeta^*$ are uniformly
distributed in $(m_1,m_2)$,
and $A^*$ is a randomly chosen subset of $\{1,\ldots,p\}$ with $|A^*|=K<n$.
The responses $\ln(T_{i})=\bx_i^T\bbeta+\epsilon_i$, where $\epsilon_i$'s are independently drawn from the normal distribution $N(0,\sigma^2)$.
The censoring variable $C_i$ is generated independently from the uniform distribution $U(0,\eta)$ as in Sect. \ref{sec:7}.
All the simulation results reported below are based on 100 independent replications.
\begin{figure*}
  \centering
   \includegraphics[width=0.75\textwidth]{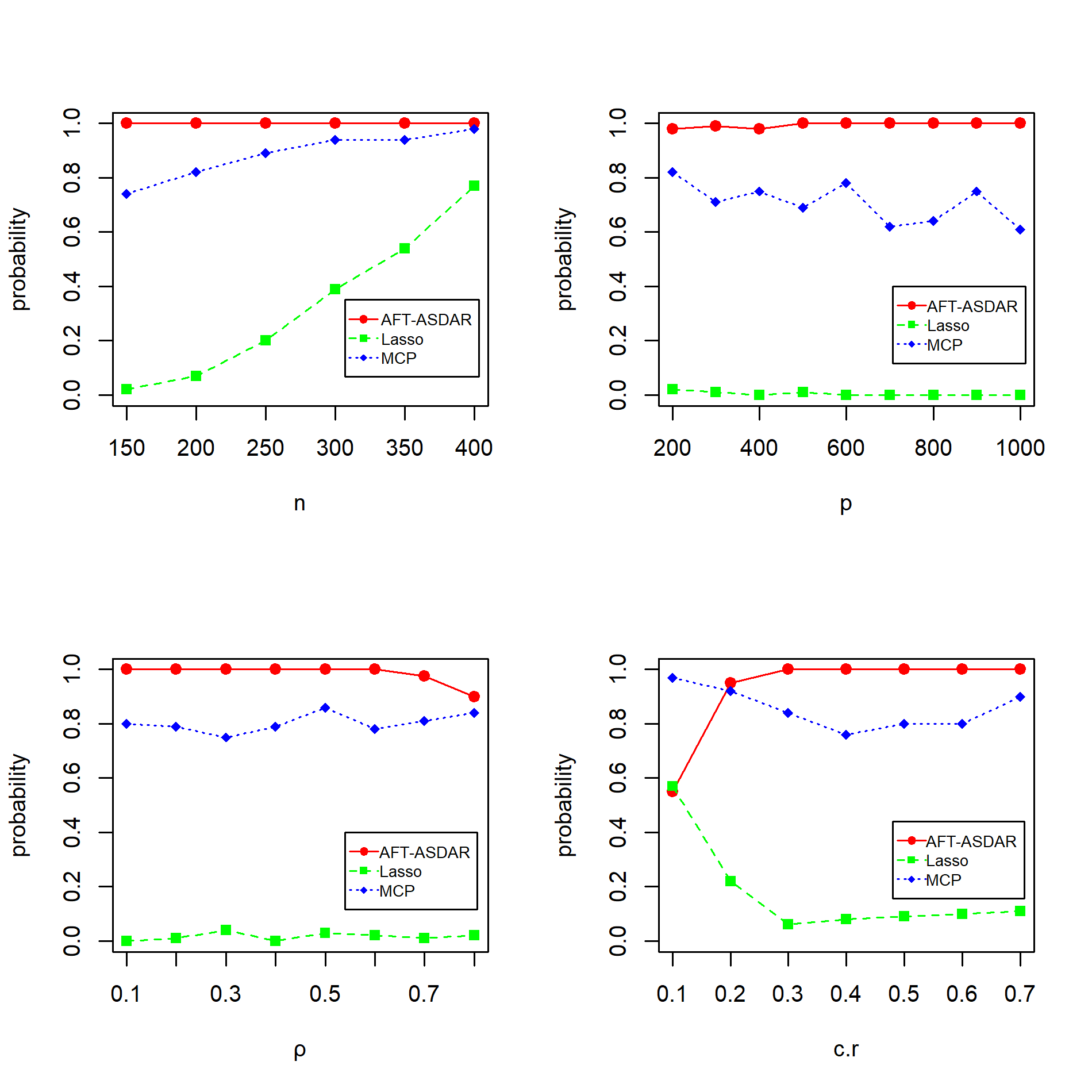}
  \caption{The numerical results of the influence of sample size $n$ (top left panel), variable dimension $p$ (top right panel), correlation $\rho$ (bottom left panel) and censoring rate $c.r$ (bottom right panel) on the probability of exact recovery of the true support sets}
  \label{fig:1}
\end{figure*}
\subsubsection{Influence of the sample size $n$}
\label{sec:9}
We set $n=100:50:400$, $p=550$, $K=6$, $R=10$, $\sigma=1$, $c.r=0.3$ and $\rho=0.3$ in the data generating models.
The top left panel of Fig. \ref{fig:1} shows the influence of the sample size $n$ on the percentage of exact recovery of $A^*$ based
on 100 replications.
In these examples,
AFT-ASDAR tends to have the percentage of recovery close to 100\%, while the percentage of Lasso is significantly less than 100\%, and the percentage of MCP is less than 100\% except when the sample size $n=400$.
\subsubsection{ Influence of the variable dimension $p$}
\label{sec:10}
We set $n=100$, $p=200:100:1000$, $K=6$, $R=10$, $\sigma=1$, $c.r=0.3$ and $\rho=0.3$ in the models.
The top right panel of Fig. \ref{fig:1} shows the influence of the variable dimension  $p$ on the percentage of exact recovery of $A^*$.
The percentage of AFT-ASDAR is always close to 100\% as the variable dimension $p$ increases, but those of both Lasso and MCP are always less than 100\%. These results suggest that AFT-ASDAR performs better in selecting variables with an increasing variable dimension $p$.
\subsubsection{Influence of the correlation $\rho$}
\label{sec:11}
We set $n=150$, $p=500$, $K=6$, $R=10$, $\sigma=1$, $c.r=0.3$ and $\rho=0.1:0.1:0.8$. The bottom left panel of Fig. \ref{fig:1} shows the influence of the the correlation $\rho$ on the percentage of exact recovery of $A^*$.
AFT-ASDAR has nearly 100\% probability in support recovery except when $\rho>0.6$.
When  $\rho=0.8$, the recovery percentage of
AFT-ASDAR is smaller than but still comparable with MCP.
\subsubsection{Influence of the censoring rate $c.r$}
\label{sec:12}
 We set $n=200$, $p=500$, $K=6$, $R=10$, $\sigma=1$, $c.r=0.1:0.1:0.7$
and $\rho=0.3$ to generate the data.
The bottom right panel of Fig. \ref{fig:1} shows the influence of  the censoring rate $c.r$ on the probability of exact recovery of $A^*$.
As the censoring rate $c.r$ increases, the percentage of recovery  of AFT-ASDAR is stable and remains close to 1, while the recovery percentages  of Lasso and MCP are less than 1.
\subsection{Number of iterations}
\label{sec:13}
To examine the convergence properties of
of AFT-SDAR, we conduct  simulations to obtain the average number of iterations of AFT-SDAR with $K$=$T$ in Algorithm \ref{alg:1}.
 We generate the data in the same way as described in
 Section \ref{sec:8}.
Figure \ref{fig:2} shows the average number of iterations of AFT-SDAR \textcolor{black}{with $\tau=1$}
based on 100 independent replications on data set:
$n=500$, $p=1000$, $K= 2 : 2 : 50$, $R=3$, $\sigma=1$, $c.r=0.3$, $\rho=0.3$.
\begin{figure}
  \resizebox{1\hsize}{!}
  {\includegraphics{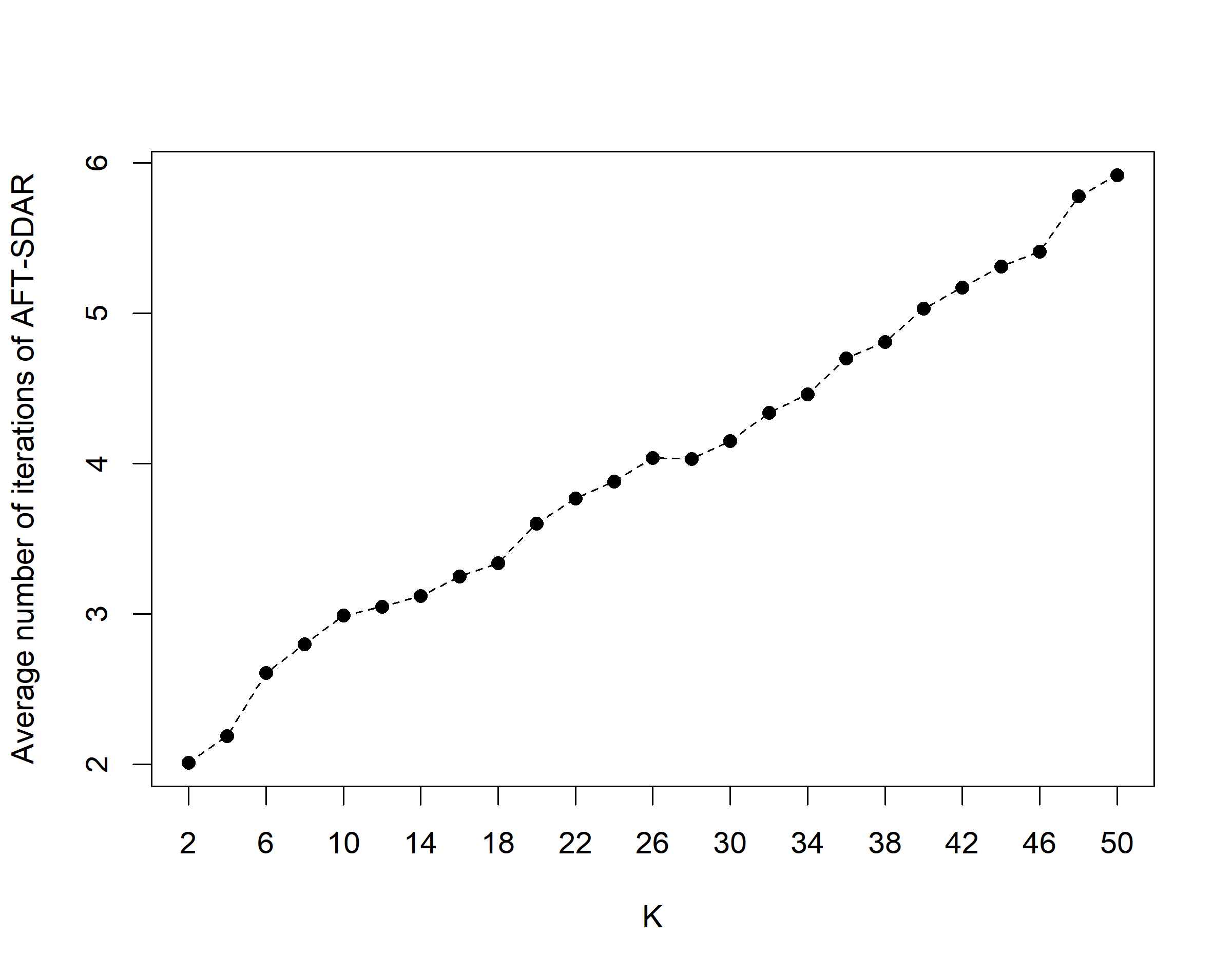}}
  \caption{The average number of iterations of AFT-SDAR as K increases}
  \label{fig:2}
\end{figure}

As shown in  Fig. \ref{fig:2}, the average number of iterations of the AFT-SDAR algorithm increases as the number of important variables $K$ increases from 2 to 50. This is expected since it will take more iterations for the algorithm to converge when the model size increases.
However, even when $K=50$,  it only take six iterations for the algorithm to converge. This shows  that AFT-SDAR has fast convergence in the simulation models considered here.
\subsection{Real data exemple}
\label{sec:14}
In this section, we illustrate the proposed approach by analyzing  the breast cancer data set nki70 from the study of
\cite{van2002gene}.
The nki70 data set includes 144 lymph node positive breast cancer patients on metastasis-free survival, 5 clinical risk factors, and gene expression measurements of 70 genes found to be prognostic for metastasis-free survival in an earlier study, and the censoring rate is about $66.67\%$.
We fit this  data set with the AFT model. Further, we  compare the estimation of the proposed approaches with that of Lasso and MCP. We set $T=0.5*n/\log(n)$ in AFT-SDAR, and implement AFT-ASDAR with $Q=0.5*n/\log(n)$.
Set $\tau=0.01$ in AFT-SDAR and AFT-ASDAR.
The results are showed in Table \ref{tab:2}.

In Table \ref{tab:2}, AFT-SDAR and AFT-ASDAR yield the same results, that is, they select the same set of genes and give the same estimated regression coefficients.
Lasso selects the largest number of genes, and MCP selects  the fewest number of genes.  The  coefficients of the common selected genes for these four methods have same sign. Especially,  AFT-SDAR and AFT-ASDAR yield similar values of the  estimated coefficients to those of Lasso for  genes SLC2A3 and C20orf46, and yield the similar value of the estimated coefficient with MCP for gene MMP9.
\begin{table}
\caption{The estimation results of nki70}
\label{tab:2}
\begin{center}
\scalebox{0.9}{
\begin{tabular}{lllllllll}
\hline\noalign{\smallskip}
Gene name&Number&Lasso&MCP&AFT-SDAR& AFT-ASDAR& \\
\noalign{\smallskip}\hline\noalign{\smallskip}
ALDH4A1&6&-1.37&-&-3.01&-3.01\\
DIAPH3.2&12&-&-&1.50&1.50\\
C16orf61&14&-&-&-1.36&-1.36\\
EXT1&16&1.86&-&3.89&3.89\\
FLT1&17&0.13&-&1.47&1.47\\
GNAZ&18&0.09&-&-&-\\
MMP9&20&-2.57&-3.48&-3.73&-3.73\\
CDC42BPA&27&-&-&1.77&1.77\\
GSTM3&30&-0.72&-&-0.93&-0.93\\
PECI&36&-&-&0.99&0.99\\
MTDH&37&-0.77&-&-1.14&-1.14\\
Contig40831\_RC&38&-0.03&-0.10&-&-\\
SLC2A3&47&1.23&2.48&1.20&1.20\\
RFC4&50&-&-&-1.64&-1.64\\
CDCA7&51&-0.35&-&-&-\\
AP2B1&55&0.26&-&-&-\\
PALM2.AKAP2&62&0.47&-&-&-\\
LGP2&63&0.13&-&0.85&0.85\\
CENPA&66&-0.78&-0.59&-&-\\
C20orf46&70&-0.89&-&-0.85&-0.85\\
\noalign{\smallskip}\hline
\end{tabular}
}
\end{center}
\end{table}
\section{Conclusion}
\label{sec:15}
In this paper, we consider the $\ell_0$-penalized method for estimation and variable selection in the high-dimensional AFT models.
We extend the SDAR algorithm for the linear regression
to the AFT model with censored survival data based on a weighted least squares criterion.
The proposed AFT-SDAR algorithm is a constructive approach for approximating
$\ell_0$-penalized weighted least squares solutions. In theoretical analysis, we establish
$\ell_\infty$  nonasymptotic error bounds for the solution sequence generated by AFT-SDAR algorithm  {under appropriate conditions weaker than  those in the existing works on nonconvex penalized regressions
\citep{Zhangzhang,huang2018constructive,mw2019},
and the key condition only relies on the identifiability of the AFT model}.
We also study the oracle support recovery property of AFT-SDAR.
Simulation studies and real data analysis demonstrate superior performance of the AFT-SDAR in terms of relative estimation error, support recovery and computational efficiency in comparison with
the lasso and MCP methods.
Therefore, AFT-SDAR can be a useful tool in addition to the existing methods for analyzing high-dimensional censored survival data.

It would be interesting to apply the proposed method to other
important survival analysis models such as the Cox model.
For the Cox model, we can consider $\ell_0$-penalized partial likelihood criterion.
Conceptually, the computational algorithm can be developed similarly based on the idea of support detection and root finding. However, the
theoretical analysis of the convergence properties of the solution
sequence is more challenging if the loss function is not quadratic and requires further work.


\section*{Appendix}
\begin{appendix}
\section{Proof of Lemma \ref{L1}}
\label{lemma1proof}
\setcounter{equation}{0}
\def\theequation{A.\arabic{equation}}
\begin{proof}
Let $\bar{L}_{\lambda}(\boldeta)=\frac{1}{2n}\|\bar{\bY}-\bar{\bX}\boldeta\|_2^2+\lambda\|\boldeta\|_0$, and
$\widetilde{L}_{\lambda}(\bbeta)=\frac{1}{2n}\sum_{i=1}^{n}w_{(i)}\big{(}Y_{(i)}-\bx^{T}_{(i)}\bbeta\big{)}^2+\lambda\|\bbeta\|_0$.
Suppose $\boldeta^{\diamond}$ is a minimizer of $\bar{L}_{\lambda}$, then
{
\begin{equation*}
\begin{split}
&\eta_{i}^{\diamond}\in \underset{t\in \mathbb{R}}{\mbox{argmin}}~\bar{L}_{\lambda}(\eta^{\diamond}_{1},\ldots,\eta^{\diamond}_{i-1},t,\eta^{\diamond}_{i+1},\ldots,\eta^{\diamond}_{p})\\
\Rightarrow &\eta_{i}^{\diamond}\in \underset{t\in \mathbb{R}}{\mbox{argmin}}~\frac{1}{2n}\left\|\bar{\bX}\boldeta^{\diamond}-\bar{\bY}+(t-\eta^\diamond_{i})\bar{\bx}_{i}\right\|_2^2+\lambda\|t\|_0\\
\Rightarrow &\eta_{i}^{\diamond}\in \underset{t\in \mathbb{R}}{\mbox{argmin}}~\frac{1}{2}(t-\eta^\diamond_{i})^2+(t-\eta^{\diamond}_{i})\bar{\bx}^{T}_{i}(\bar{\bX}\boldeta^{\diamond}-\bar{\bY})/n+\lambda\|t\|_0\\
\Rightarrow &\eta_{i}^{\diamond}\in \underset{t\in \mathbb{R}}{\mbox{argmin}}~\frac{1}{2}[t-(\eta^\diamond_{i}+\bar{\bx}^{T}_{i}(\bar{\bY}-\bar{\bX}\boldeta^{\diamond})/n)]^2+\lambda\|t\|_0.\\
\end{split}
\end{equation*}
}
Let $\bd^{\diamond}=\bar{\bX}^{T}(\bar{\bY}-\bar{\bX}\boldeta^{\diamond})/n$. By the definition of $H_{\lambda}(\cdot)$ in \eqref{eq3}, we have
\begin{equation*}
\eta^{\diamond}_{i}=H_{\lambda}(\eta^{\diamond}_{i}+d^{\diamond}_{i})\quad \quad \mbox{for}\quad i =1,...,p,
\end{equation*}
which shows \eqref{eq2} holds.

Conversely, if $\boldeta^{\diamond}$ and $\bd^{\diamond}$ satisfy \eqref{eq2}, then we will show that $\boldeta^{\diamond}$ is a local minimizer of \eqref{eq1}, and $\bbeta^{\diamond}=D\cdot\boldeta^{\diamond}$ is a local minimizer of \eqref{eq03} too.
We can assume $\bh$ is small enough and $\|\bh\|_{\infty}<\sqrt{2\lambda}$. Then we will show $\bar{L}_{\lambda}(\boldeta^{\diamond}+\bh)\geq \bar{L}_{\lambda}(\boldeta^{\diamond})$ in two case respectively.\\
\textbf{Case1:} $\bh_{I^{\diamond}}\neq0$.
\begin{equation*}
\|\boldeta^{\diamond}+\bh\|_{0}=\|\boldeta_{A^\diamond}^{\diamond}+\bh_{A^\diamond}\|_{0}+\|\bh_{I^\diamond}\|_{0},
\end{equation*}
\begin{equation*}
\lambda\|\boldeta^{\diamond}+\bh\|_{0}-\lambda\|\boldeta^{\diamond}\|_{0}=\lambda\|\boldeta_{A^\diamond}^{\diamond}+\bh_{A^\diamond}\|_{0}+\lambda\|\bh_{I^\diamond}\|_{0}-\lambda\|\boldeta_{A^\diamond}^{\diamond}\|_{0}.
\end{equation*}
Because $|\boldeta^{\diamond}_{i}|\geq\sqrt{2\lambda}$  for  $i\in{A^\diamond}$  and  $\|\bh\|_{\infty}<\sqrt{2\lambda}$, we have \begin{equation*}
\lambda\|\boldeta_{A^\diamond}^{\diamond}+\bh_{A^\diamond}\|_{0}-\lambda\|\boldeta_{A^\diamond}^{\diamond}\|_{0}=0,
\end{equation*}
\begin{equation*}
\lambda\|\boldeta^{\diamond}+\bh\|_{0}-\lambda\|\boldeta^{\diamond}\|_{0}=\lambda\|\bh_{I^\diamond}\|_{0}>\lambda.
\end{equation*}
Therefore, we get
\begin{equation*}
\begin{split}
&\bar{L}_{\lambda}(\boldeta^{\diamond}+\bh)-\bar{L}_{\lambda}(\boldeta^{\diamond})\\
&=\frac{1}{2n}\left\|\bar{\bY}-\bar{\bX}(\boldeta^{\diamond}+\bh)\right\|_2^2-\frac{1}{2n}\left\|\bar{\bY}-\bar{\bX}\boldeta^{\diamond}\right\|_2^2
+\lambda\|\bh_{I^\diamond}\|_{0}\\
&=\frac{1}{2n}\left[\left\|\bar{\bX}\bh\right\|_2^2-2(\bar{\bY}-\bar{\bX}\boldeta^{\diamond})^{T}\bar{\bX}\bh\right]+\lambda\|\bh_{I^\diamond}\|_{0}\\
&\geq \lambda-\langle \bd^{\diamond},\bh\rangle.
\end{split}
\end{equation*}
The last inequality $\lambda-\langle \bd^{\diamond},\bh\rangle\geq0$ holds for any small enough vector \bh, so we obtain $\bar{L}_{\lambda}(\boldeta^{\diamond}+\bh)-\bar{L}_{\lambda}(\boldeta^{\diamond})\geq0$.\\
\textbf{Case2:} $\bh_{I^{\diamond}}=0$. \begin{equation*}
\lambda\|\boldeta^{\diamond}+\bh\|_{0}-\lambda\|\boldeta^{\diamond}\|_{0}=\lambda\|\boldeta_{A^\diamond}^{\diamond}+\bh_{A^\diamond}\|_{0}-\lambda\|\boldeta_{A^\diamond}^{\diamond}\|_{0}.
\end{equation*}
As $|\eta^{\diamond}_{i}|\geq\sqrt{2\lambda}$  for  $i\in{A^\diamond}$ and $\|\bh_{A^\diamond}\|_{\infty}<\sqrt{2\lambda}$, then we have \begin{equation*}
\lambda\|\boldeta^{\diamond}+\bh\|_{0}-\lambda\|\boldeta^{\diamond}\|_{0}=\lambda\|\boldeta_{A^\diamond}^{\diamond}+\bh_{A^\diamond}\|_{0}-\lambda\|\boldeta_{A^\diamond}^{\diamond}\|_{0}=0.
\end{equation*}
Due to $\bd^{\diamond}_{A^\diamond}=\bar{\bX}^{T}_{A^\diamond}(\bar{\bY}-\bar{\bX}_{A^\diamond}\boldeta^{\diamond}_{A^\diamond})/n=0$, then we can get
\begin{equation*}
\boldeta^{\diamond}_{A^\diamond}\in \underset{\boldeta_{A^\diamond}}{\mbox{argmin}}~\frac{1}{2n}\left\|\bar{\bX}_{A^\diamond}\boldeta_{A^\diamond}-\bar{\bY}\right\|_2^2.
\end{equation*}
Thus, we conclude that
\begin{equation*}
\begin{split}
&\bar{L}_{\lambda}(\boldeta^\diamond+\bh)-\bar{L}_{\lambda}(\boldeta^\diamond)\\
&=\frac{1}{2n}\left\|\bar{\bY}-\bar{\bX}(\boldeta^{\diamond}+\bh)\right\|_2^2-\frac{1}{2n}\left\|\bar{\bY}-\bar{\bX}\boldeta^{\diamond}\right\|_2^2\\
&=\frac{1}{2n}\left\|\bar{\bY}-\bar{\bX}_{A^\diamond}(\boldeta^{\diamond}_{A^\diamond}+\bh_{A^\diamond})\right\|_2^2-
\frac{1}{2n}\left\|\bar{\bY}-\bar{\bX}_{A^\diamond}\boldeta^{\diamond}_{A^\diamond}\right\|_2^2\\
&\geq0.
\end{split}
\end{equation*}
In summary, $\boldeta^{\diamond}$ is a local minimizer of $\bar{L}_{\lambda}$.
Let $\bar{\bh}=D\cdot \bh$,
then $\widetilde{L}_{\lambda}(\bbeta^\diamond+\bar{\bh})-\widetilde{L}_{\lambda}(\bbeta^\diamond)\geq0$ holds if the vector $\bh$  is sufficiently small, thus $\bbeta^\diamond$ is also a local
minimizer of \eqref{eq03}.
\qed
\end{proof}
\begin{lemma}\label{L9.0}
There exists  constants $0<L\leq U<\infty$ with $0<L\leq\frac{\sigma_{(\min,2T)}}{n\sqrt{2T}}$ and $\frac{\|\bar{\bX}\|_2^2}{n}\leq U<\infty$ such that for all different p-dimensional vectors $\boldeta_{1}$ and $\boldeta_{2}$ with $\|\boldeta_{1}-\boldeta_{2}\|_{0}\leq 2T$,
\begin{equation}\label{UL}
0<L\leq\frac{(\boldeta_{1}-\boldeta_{2})^{T}\cdot\bar{\bX}^T\bar{\bX}\cdot(\boldeta_{1}-\boldeta_{2})}{n\|\boldeta_{1}-\boldeta_{2}\|_1\|\boldeta_{1}-\boldeta_{2}\|_{\infty}}\leq U<\infty.
\end{equation}
\end{lemma}
\begin{proof}
Since $\|\boldeta_{1}-\boldeta_{2}\|_1\|\boldeta_{1}-\boldeta_{2}\|_{\infty}\geq\|\boldeta_{1}-\boldeta_{2}\|_2^2$,
we have
\begin{equation*}
\frac{(\boldeta_{1}-\boldeta_{2})^{T}\cdot\bar{\bX}^T\bar{\bX}\cdot(\boldeta_{1}-\boldeta_{2})}{n\|\boldeta_{1}-\boldeta_{2}\|_1\|\boldeta_{1}-\boldeta_{2}\|_{\infty}}\leq\frac{\|\bar{\bX}\|^2_2}{n}.
\end{equation*}
Hence, there exist $U\in \left[\frac{\|\bar{\bX}\|_2^2}{n},\infty\right)$ such that the right hand side of \eqref{UL} holds.
Moreover, since $\|\boldeta_{1}-\boldeta_{2}\|_1\|\boldeta_{1}-\boldeta_{2}\|_{\infty}\leq\sqrt{2T}\|\boldeta_{1}-\boldeta_{2}\|_2^2$,
we have
\begin{equation*}
\frac{(\boldeta_{1}-\boldeta_{2})^{T}\cdot\bar{\bX}^T\bar{\bX}\cdot(\boldeta_{1}-\boldeta_{2})}{n\|\boldeta_{1}-\boldeta_{2}\|_1\|\boldeta_{1}-\boldeta_{2}\|_{\infty}}\geq\frac{\sigma_{(\min,2T)}}{n\sqrt{2T}}.
\end{equation*}
Thus, there exists $L \in \left(0,\frac{\sigma_{(\min,2T)}}{n\sqrt{2T}}\right]$
such that the left hand side of \eqref{UL} holds.
\qed
\end{proof}
\begin{lemma}\label{L9.1}
Assume $0<L\leq\frac{\sigma_{(\min,2T)}}{n\sqrt{2T}}$ and $\|\boldeta^{*}\|_{0}=K\leq T$. Denote $B^k = A^{k}\backslash A^{k-1}$. Then,
\begin{equation*}
\|\nabla_{B^k}\mathcal{L}_2(\boldeta^{k})\|_1\|\nabla_{B^k}\mathcal{L}_2(\boldeta^{k})\|_{\infty}\geq 2L\zeta[\mathcal{L}_2(\boldeta^k)-\mathcal{L}_2(\boldeta^*)],
\end{equation*}
where $\zeta=\frac{|B^k|}{|B^k|+|A^*\backslash A^{k-1}|}$.
\end{lemma}
\begin{proof}
Obviously, this lemma holds if $A^{k}=A^{k-1}$ or $\mathcal{L}_2(\boldeta^k)\leq \mathcal{L}_2(\boldeta^*)$. So, we only prove the lemma by assuming $A^{k}\neq A^{k-1}$ and $\mathcal{L}_2(\boldeta^k)>\mathcal{L}_2(\boldeta^*)$. As $0<L\leq\frac{\sigma_{(\min,2T)}}{n\sqrt{2T}}$,
 the left hand side of  \eqref{UL} holds. It implies that
\begin{equation*}
\begin{split}
&\mathcal{L}_2(\boldeta^{*})-\mathcal{L}_2(\boldeta^k)-\langle\nabla \mathcal{L}_2(\boldeta^k),{\boldeta^*-\boldeta^k}\rangle\geq\frac{L}{2}\left\|\boldeta^*-\boldeta^k\right\|_1\left\|\boldeta^*-\boldeta^k\right\|_{\infty}.
\end{split}
\end{equation*}
Hence,
\begin{equation*}
\begin{split}
&-\langle\nabla \mathcal{L}_2(\boldeta^k),{\boldeta^*-\boldeta^k}\rangle\\
&=\langle\nabla \mathcal{L}_2(\boldeta^k),-\boldeta^*\rangle\\
&\geq\frac{L}{2}\left\|\boldeta^*-\boldeta^k\right\|_1\left\|\boldeta^*-\boldeta^k\right\|_{\infty}+\mathcal{L}_2(\boldeta^k)-\mathcal{L}_2(\boldeta^{*})\\
&\geq \sqrt{2L}\sqrt{\left\|\boldeta^*-\boldeta^k\right\|_1\left\|\boldeta^*-\boldeta^k\right\|_{\infty}}\sqrt{\mathcal{L}_2(\boldeta^k)-\mathcal{L}_2(\boldeta^{*})}.\\
\end{split}
\end{equation*}
By the definition of $A^{k}$,  $B^k$ contains the first $|B^k|$-largest elements (in absolute value) of $\nabla \mathcal{L}_2(\boldeta^k)$ and
\begin{equation*}
\mbox{supp}(\nabla \mathcal{L}_2(\boldeta^k))\bigcap \mbox{supp}(\boldeta^*)=A^{*}\backslash A^{k-1}.
\end{equation*}
 Thus,  we get
\begin{equation*}
\begin{split}
&\langle\nabla \mathcal{L}_2(\boldeta^k),-\boldeta^*\rangle\\
&\leq\frac{1}{\sqrt{\zeta}}\|\nabla_{B^k}\mathcal{L}_2(\boldeta^k)\|_2\|\boldeta_{A^*\backslash A^{k-1}}^{*}\|_2\\
&=\frac{1}{\sqrt{\zeta}}\|\nabla_{B^k}\mathcal{L}_2(\boldeta^k)\|_2\|(\boldeta^{*}-\boldeta^k)_{A^*\backslash A^{k-1}}\|_2\\
&\leq\frac{1}{\sqrt{\zeta}}\|\nabla_{B^k}\mathcal{L}_2(\boldeta^k)\|_2\|\boldeta^{*}-\boldeta^k\|_2\\
&\leq\frac{1}{\sqrt{\zeta}}\sqrt{\|\nabla_{B^k}\mathcal{L}_2(\boldeta^k)\|_1\|\nabla_{B^k}\mathcal{L}_2(\boldeta^k)\|_{\infty}}
\cdot\sqrt{\|\boldeta^{*}-\boldeta^k\|_1\|\boldeta^{*}-\boldeta^k\|_{\infty}}.
\end{split}
\end{equation*}
Therefore,
{
\begin{equation*}
\sqrt{2L}\sqrt{\mathcal{L}_2(\boldeta^k)-\mathcal{L}_2(\boldeta^{*})}\leq\frac{1}{\sqrt{\zeta}}\sqrt{\|\nabla_{B^k}\mathcal{L}_2(\boldeta^k)\|_1\|\nabla_{B^k}\mathcal{L}_2(\boldeta^k)\|_{\infty}}.
\end{equation*}
}
In summary,
\begin{equation*}
\|\nabla_{B^k}\mathcal{L}_2(\boldeta^k)\|_1\|\nabla_{B^k}\mathcal{L}_2(\boldeta^k)\|_{\infty}\geq 2L\zeta[\mathcal{L}_2(\boldeta^k)-\mathcal{L}_2(\boldeta^*)].
\end{equation*}
\qed
\end{proof}
\begin{lemma}\label{L9.2}
Assume $\tau<\frac{1}{ \sqrt{T}U}$  with $\frac{\|\bar{\bX}\|_2^2}{n}\leq U<\infty$
and $0<L\leq\frac{\sigma_{(\min,2T)}}{n\sqrt{2T}}$,
and set $K\leq T$ in Algorithm \ref{alg:1}.
 Then before Algorithm \ref{alg:1} terminates, the following inequality holds for all $k\geq0$:
\begin{equation*}
\mathcal{L}_2(\boldeta^{k+1})-\mathcal{L}_2(\boldeta^*)\leq\xi[\mathcal{L}_2(\boldeta^k)-\mathcal{L}_2(\boldeta^*)],
\end{equation*}
where $\xi=1-\frac{2\tau L(1-\tau \sqrt{T}U)}{ \sqrt{T}(1+K)}\in(0,1)$.
\end{lemma}
\begin{proof}
Let $\bb^{k}=\boldeta^k-\tau \nabla \mathcal{L}_2(\boldeta^{k})$.
The right hand side of  \eqref{UL} implies
\begin{equation*}
\begin{split}
&\mathcal{L}_2(\bb^{k+1}|_{A^{k+1}})-\mathcal{L}_2(\boldeta^{k+1})\\
&\leq\langle\nabla \mathcal{L}_2(\boldeta^{k+1}),\bb^{k+1}|_{A^{k+1}}-\boldeta^{k+1}\rangle+\frac{U}{2}\left\|\bb^{k+1}|_{A^{k+1}}-\boldeta^{k+1}\right\|_1\left\|\bb^{k+1}|_{A^{k+1}}-\boldeta^{k+1}\right\|_{\infty}.
\end{split}
\end{equation*}
On one hand,
\begin{equation*}
\begin{split}
&\langle\nabla \mathcal{L}_2(\boldeta^{k+1}), \bb^{k+1}|_{A^{k+1}}-\boldeta^{k+1}\rangle\\
&=\langle\nabla \mathcal{L}_2(\boldeta^{k+1}), \bb^{k+1}|_{A^{k+1}}\rangle\\
&=\langle\nabla_{A^{k+1}}\mathcal{L}_2(\boldeta^{k+1}), \bb_{A^{k+1}}^{k+1}\rangle\\
&=\langle\nabla_{A^{k+1}\backslash A^{k}}\mathcal{L}_2(\boldeta^{k+1}), \bb_{A^{k+1}\backslash A^{k}}^{k+1}\rangle.
\end{split}
\end{equation*}
Furthermore, we also have
\begin{equation*}
\begin{split}
&\big{\|}\bb^{k+1}|_{A^{k+1}}-\boldeta^{k+1}\big{\|}_1\\
&=\big{\|}\bb^{k+1}|_{A^{k+1}\backslash A^{k}}+\bb^{k+1}|_{A^{k+1}\bigcap A^{k}}-\boldeta^{k+1}|_{A^{k+1}\bigcap A^{k}}-\boldeta^{k+1}|_{A^{k}\backslash A^{k+1}}\big{\|}_1\\
&=\big{\|}\bb_{A^{k+1}\backslash A^{k}}^{k+1}\|_1+\|\bb_{A^{k+1}\bigcap A^{k}}^{k+1}-\boldeta_{A^{k+1}\bigcap A^{k}}^{k+1}\|_1+\|\boldeta_{A^{k}\backslash A^{k+1}}^{k+1}\big{\|}_1\\
&=\big{\|}\bb_{A^{k+1}\backslash A^{k}}^{k+1}\|_1+\|\boldeta_{A^{k}\backslash A^{k+1}}^{k+1}\big{\|}_1,
\end{split}
\end{equation*}
and
\begin{equation*}
\begin{split}
&\big{\|}\bb^{k+1}|_{A^{k+1}}-\boldeta^{k+1}\big{\|}_{\infty}\\
&=\big{\|}\bb^{k+1}|_{A^{k+1}\backslash A^{k}}+\bb^{k+1}|_{A^{k+1}\bigcap A^{k}}-\boldeta^{k+1}|_{A^{k+1}\bigcap A^{k}}-\boldeta^{k+1}|_{A^{k}\backslash A^{k+1}}\big{\|}_{\infty}\\
&=\big{\|}\bb_{A^{k+1}\backslash A^{k}}^{k+1}\|_{\infty}\bigvee\|\boldeta_{A^{k}\backslash A^{k+1}}^{k+1}\big{\|}_{\infty},
\end{split}
\end{equation*}
where $c\bigvee d=\max\left\{c,d\right\}$.
On the other hand, by the definition of $A^k$, $A^{k+1}$ and $\boldeta^{k+1}$, we know that
\begin{equation*}
|A^{k}\backslash A^{k+1}|=|A^{k+1}\backslash A^{k}|,~ \bb_{A^{k}\backslash A^{k+1}}^{k+1}=\boldeta_{A^{k}\backslash A^{k+1}}^{k+1}.
\end{equation*}
By the definition of $A^{k+1}$, we conclude that
\begin{equation*}
\big{\|}\bb_{A^{k}\backslash A^{k+1}}^{k+1}\big{\|}_1=\big{\|}\boldeta_{A^{k}\backslash A^{k+1}}^{k+1}\big{\|}_1\leq \big{\|}\bb_{A^{k+1}\backslash A^{k}}^{k+1}\big{\|}_1,
\end{equation*}
\begin{equation*}
\big{\|}\bb_{A^{k+1}\backslash A^{k}}^{k+1}\|_{\infty}\bigvee\|\boldeta_{A^{k}\backslash A^{k+1}}^{k+1}\big{\|}_{\infty}=\big{\|}\bb_{A^{k+1}\backslash A^{k}}^{k+1}\|_{\infty}.
\end{equation*}
Due to $-\nabla_{A^{k+1}\backslash A^{k}}\mathcal{L}_2(\boldeta^{k+1})=\frac{1}{\tau}\bb_{A^{k+1}\backslash A^{k}}^{k+1}$ and $\tau<\frac{1}{ \sqrt{T}U}$, hence we have
\begin{equation*}
\begin{split}
&\mathcal{L}_2(\bb^{k+1}|_{A^{k+1}})-\mathcal{L}_2(\boldeta^{k+1})\\
&\leq\langle\nabla_{A^{k+1}\backslash A^{k}}\mathcal{L}_2(\boldeta^{k+1}),\bb_{A^{k+1}\backslash A^{k}}^{k+1}\rangle+U\big{\|}\bb_{A^{k+1}\backslash A^{k}}^{k+1}\big{\|}_1\big{\|}\bb_{A^{k+1}\backslash A^{k}}^{k+1}\big{\|}_{\infty}\\
&\leq-\left(\frac{\tau}{ \sqrt{T}}-U\tau^2\right)\big{\|}\nabla_{A^{k+1}\backslash A^{k}}\mathcal{L}_2(\boldeta^{k+1})\big{\|}_1\cdot\big{\|}\nabla_{A^{k+1}\backslash A^{k}}\mathcal{L}_2(\boldeta^{k+1})\big{\|}_{\infty}.
\end{split}
\end{equation*}
By the definition of $\boldeta^{k+1}$, we  get
\begin{equation*}
\begin{split}
&\mathcal{L}_2(\boldeta^{k+1})-\mathcal{L}_2(\boldeta^{k})\\
&\leq \mathcal{L}_2(\bb^{k}|_{A^{k}})-\mathcal{L}_2(\boldeta^{k})\\
&\leq-\left(\frac{\tau}{ \sqrt{T}}-U\tau^2\right)\big{\|}\nabla_{B^k}\mathcal{L}_2(\boldeta^{k})\big{\|}_1\big{\|}\nabla_{B^k}\mathcal{L}_2(\boldeta^{k})\big{\|}_{\infty}.
\end{split}
\end{equation*}
Moreover, $\frac{|A^*\backslash A^{k-1}|}{|B^{k}|}\leq K$.
By Lemma $\ref{L9.1}$, we have
\begin{equation*}
\mathcal{L}_2(\boldeta^{k+1})-\mathcal{L}_2(\boldeta^{k})\leq-\frac{2\tau L(1-\tau \sqrt{T} U)}{ \sqrt{T}(1+K)}[\mathcal{L}_2(\boldeta^{k})-\mathcal{L}_2(\boldeta^{*})].
\end{equation*}
Therefore, we obtain that
 \begin{equation*}
\mathcal{L}_2(\boldeta^{k+1})-\mathcal{L}_2(\boldeta^*)\leq\xi[\mathcal{L}_2(\boldeta^k)-\mathcal{L}_2(\boldeta^*)],
\end{equation*}
where
$\xi=1-\frac{2\tau L(1-\tau \sqrt{T}U)}{ \sqrt{T}(1+K)}\in(0,1)$.
\qed
\end{proof}
\begin{lemma}\label{L9.3}
Assume
$0<L\leq\frac{\sigma_{(\min,2T)}}{n\sqrt{2T}}$ and $\frac{\|\bar{\bX}\|_2^2}{n}\leq U<\infty$.
Suppose that $\boldeta^{*}$ is an arbitrary sparse vector with $\|\boldeta^{*}\|_{0}=K\leq T$, $\|\boldeta^k\|_{0}=T$ and $\mathcal{L}_2(\boldeta^{k+1})-\mathcal{L}_2(\boldeta^{*})\leq \xi[\mathcal{L}_2(\boldeta^{k})-\mathcal{L}_2(\boldeta^{*})]$ for all $k\geq0$,  where $0<\xi<1$. Then,
\begin{equation}\label{eq001}
\begin{split}
\|\boldeta^{k}-\boldeta^{*}\|_{\infty}\leq&
\sqrt{(1+U/L)\cdot(K+T)}(\sqrt{\xi})^k\|\boldeta^0-\boldeta^{*}\|_{\infty}+\frac{2}{L}\|\nabla \mathcal{L}_2(\boldeta^{*})\|_{\infty},
\end{split}
\end{equation}
\end{lemma}
\begin{proof}
If $\|\boldeta^k-\boldeta^*\|_{\infty}< \frac{2\|\nabla \mathcal{L}_2(\boldeta^*)\|_{\infty}}{L}$, then \eqref{eq001} holds, so we only concentrate on the case that $\|\boldeta^k-\boldeta^*\|_{\infty}\geq \frac{2\|\nabla
\mathcal{L}_2(\boldeta^*)\|_{\infty}}{L}$.
On one hand,
by the left hand side of \eqref{UL},
we have
{
\begin{equation*}
\begin{split}
&\mathcal{L}_2(\boldeta^{k})-\mathcal{L}_2(\boldeta^{*})\\
&\geq\langle\nabla \mathcal{L}_2(\boldeta^*),\boldeta^{k}-\boldeta^{*}\rangle+\frac{L}{2}\left\|\boldeta^{k}-\boldeta^{*}\right\|_1\left\|\boldeta^{k}-\boldeta^{*}\right\|_{\infty}\\
&\geq-\|\nabla \mathcal{L}_2(\boldeta^{*})\|_{\infty}\|\boldeta^{k}-\boldeta^{*}\|_1+\frac{L}{2}\left\|\boldeta^{k}-\boldeta^{*}\right\|_1\left\|\boldeta^{k}-\boldeta^{*}\right\|_{\infty}.
\end{split}
\end{equation*}
}
Furthermore,
{
\begin{equation*}
(\|\boldeta^k-\boldeta^*\|_1-\|\boldeta^k-\boldeta^*\|_{\infty})\left(\frac{L}{2}\big{\|}\boldeta^k-\boldeta^*\big{\|}_{\infty}-\big{\|}\nabla \mathcal{L}_2(\boldeta^*)\big{\|}_{\infty}\right)\geq 0.
\end{equation*}
}
 Then, we can get
{
\begin{equation*}
\frac{L}{2}\big{\|}\boldeta^k-\boldeta^*\big{\|}_{\infty}^2-\|\nabla \mathcal{L}_2(\boldeta^*)\|_{\infty}\|\boldeta^k-\boldeta^*\|_{\infty}-[\mathcal{L}_2(\boldeta^k)-\mathcal{L}_2(\boldeta^*)]\leq0,
\end{equation*}
}
which is one univariate quadratic inequality about $\|\boldeta^{k}-\boldeta^{*}\|$.
Therefore, we have
\begin{equation*}
\begin{split}
&\|\boldeta^{k}-\boldeta^{*}\|_{\infty}\leq\frac{\|\nabla \mathcal{L}_2(\boldeta^{*})\|_{\infty}+\sqrt{\left\|\nabla \mathcal{L}_2(\boldeta^{*})\right\|_{\infty}^2+2L[\mathcal{L}_2(\boldeta^{k})-\mathcal{L}_2(\boldeta^{*})]}}{L}.
\end{split}
\end{equation*}
Thus, we can get
\begin{equation}\label{eq002}
\begin{split}
&\|\boldeta^{k}-\boldeta^{*}\|_{\infty}\leq \sqrt{\frac{2\max\{\mathcal{L}_2(\boldeta^{k})-\mathcal{L}_2(\boldeta^{*}),0\}}{L}}+\frac{2\|\nabla \mathcal{L}_2(\boldeta^{*})\|_{\infty}}{L}.
\end{split}
\end{equation}
On the other hand, based on the right hand side of \eqref{UL},
we have
{
\begin{equation*}
\begin{split}
&\mathcal{L}_2(\boldeta^0)-\mathcal{L}_2(\boldeta^*)\\
&\leq\langle\nabla \mathcal{L}_2(\boldeta^*),\boldeta^0-\boldeta^*\rangle+\frac{U}{2}\left\|\boldeta^0-\boldeta^*\right\|_2^2\\
&\leq\langle\nabla \mathcal{L}_2(\boldeta^*),\boldeta^0-\boldeta^*\rangle+\frac{U}{2}\left\|\boldeta^0-\boldeta^*\right\|_1\left\|\boldeta^0-\boldeta^*\right\|_{\infty}\\
&\leq\|\nabla \mathcal{L}_2(\boldeta^*)\|_{\infty}\|\boldeta^0-\boldeta^*\|_1+\frac{U}{2}\left\|\boldeta^0-\boldeta^*\right\|_1\left\|\boldeta^0-\boldeta^*\right\|_{\infty}\\
&\leq (K+T)\|\boldeta^0-\boldeta^*\|_{\infty}(\|\nabla \mathcal{L}_2(\boldeta^*)\|_{\infty}+\frac{U}{2}\|\boldeta^0-\boldeta^*\|_{\infty}).
\end{split}
\end{equation*}
}
Furthermore,
{
\begin{equation*}
\begin{split}
&\mathcal{L}_2(\boldeta^{k})-\mathcal{L}_2(\boldeta^{*})\\
&\leq \xi[\mathcal{L}_2(\boldeta^{k-1})-\mathcal{L}_2(\boldeta^{*})]\\
&\leq \xi^{k}[\mathcal{L}_2(\boldeta^{0})-\mathcal{L}_2(\boldeta^{*})]\\
&\leq\xi^k(K+T)\|\boldeta^0-\boldeta^*\|_{\infty}(\|\nabla \mathcal{L}_2(\boldeta^*)\|_{\infty}+\frac{U}{2}\|\boldeta^0-\boldeta^*\|_{\infty})\\
&\leq\frac{\xi^k(K+T)(L+U)}{2}\|\boldeta^0-\boldeta^*\|_{\infty}^2.\\
\end{split}
\end{equation*}
}
Hence, by \eqref{eq002}, we have
\begin{equation*}
\begin{split}
\|\boldeta^{k}-\boldeta^{*}\|_{\infty}\leq&
\sqrt{(K+T)(1+U/L)}(\sqrt{\xi})^k\|\boldeta^0-\boldeta^*\|_{\infty}+\frac{2}{L}\|\nabla \mathcal{L}_2(\boldeta^*)\|_{\infty}.
\end{split}
\end{equation*}
This completes the proof.
\qed
\end{proof}
\begin{lemma}\label{Le}(Lemma 1 of   
\cite{huang2010variable}).
Suppose that conditions 
(C\ref{cond3})-(C\ref{cond7})
 hold, then
\begin{equation*}
\begin{split}
&E\left(\|\nabla \mathcal{L}_1(\bbeta^*)\|_\infty\right)\\
&\leq C_1\sqrt{\frac{\log(p)}{n}}\left(\sqrt{\frac{2C_2\log(p)}{n}}+\frac{4\log(2p)}{n}+C_2\right)^{\frac{1}{2}},
\end{split}
\end{equation*}
where $C_1$ and $C_2$ are two finite positive constants. In particular, when $n\gg\log(p)$,
$$E\left(\|\nabla \mathcal{L}_1(\bbeta^*)\|_\infty\right)=o(1).$$
\end{lemma}
\section{Proof of Theorem \ref{th1}}
\setcounter{equation}{0}
\def\theequation{B.\arabic{equation}}
\label{th1proof}
\begin{proof}
By Lemma $\ref{L9.2}$, we have \begin{equation*}
\mathcal{L}_2(\boldeta^{k+1})-\mathcal{L}_2(\boldeta^*)\leq\xi[\mathcal{L}_2(\boldeta^k)-\mathcal{L}_2(\boldeta^*)],
\end{equation*}
where
$\xi=1-\frac{2\tau L(1-\tau \sqrt{T}U)}{ \sqrt{T}(1+K)}\in(0,1)$.
Therefore, the conditions of Lemma $\ref{L9.3}$ are satisfied. Taking $\boldeta^0=D^{-1}\bbeta^0=0$, then we can get
\begin{equation}\label{error1}
\begin{split}
\|\boldeta^k-\boldeta^*\|_{\infty}\leq& \sqrt{(K+T)(1+U/L)}(\sqrt{\xi})^k\|\boldeta^{*}\|_{\infty}+\frac{2}{L}\|\nabla \mathcal{L}_2(\boldeta^{*})\|_{\infty}.
\end{split}
\end{equation}
By condition 
(C\ref{cond2})
and \eqref{error1}, we have
{
\begin{equation*}
\begin{split}
&\|\bbeta^k-\bbeta^*\|_\infty\\
&=\|D(\boldeta^k-\boldeta^*)\|_\infty\\
&\leq b\|\boldeta^k-\boldeta^*\|_\infty\\
&\leq b^2\sqrt{(K+T)(1+U/L)}(\sqrt{\xi})^k\|\bbeta^{*}\|_\infty+\frac{2b^2}{L}\|\nabla \mathcal{L}_1(\bbeta^{*})\|_\infty.
\end{split}
\end{equation*}
}
This completes the proof. \qed
\end{proof}
\section{Proof of Theorems \ref{th1b}}
\setcounter{equation}{0}
\def\theequation{C.\arabic{equation}}
\begin{proof}
By Lemma \ref{Le} and Markov inequality, we have
\begin{equation*}
P\left(\|\nabla \mathcal{L}_1(\bbeta^*)\|_\infty\geq\varepsilon_1\right)\leq\left(\frac{\log(p)}{n}\right)^{\frac{1}{4}},
\end{equation*}
where \begin{equation*}
\varepsilon_1= C_1\left(\frac{\log(p)}{n}\right)^{\frac{1}{4}}\left(\sqrt{\frac{2C_2\log(p)}{n}}+\frac{4\log(2p)}{n}+C_2\right)^{\frac{1}{2}}.
\end{equation*}
Then,
with probability at least
1-$\left(\frac{\log(p)}{n}\right)^{\frac{1}{4}}$,
\begin{equation}\label{errorlinf}
\begin{split}
\|\bbeta^k-\bbeta^*\|_{\infty}\leq& b^2\sqrt{(K+T)(1+U/L)}(\sqrt{\xi})^k\|\bbeta^{*}\|_{\infty}+\frac{2b^2}{L}\varepsilon_1
\end{split}
\end{equation}
This completes the proof.
\qed
\end{proof}
\section{Proof of Theorem  \ref{th2}}
\label{th2proof}
\begin{proof}
By \eqref{errorlinf} and condition 
(C\ref{cond7}),
 some algebra show that
\begin{equation*}
\begin{split}
\|\bbeta^{k}-\bbeta^{*}\|_{\infty}\leq&
b^2\sqrt{(K+T)(1+U/L)}(\sqrt{\xi})^k\|\bbeta^*\|_{\infty}+\frac{2}{3}\|\bbeta^*_{A^*}\|_{\min}< \|\bbeta^*_{A^*}\|_{\min},
\end{split}
\end{equation*}
if $k>\log_{\frac{1}{\xi}} 9 (T+K)(1+U/L)r^2b^4.$
This implies  that $A^* \subseteq A^k$.
\qed
\end{proof}
\end{appendix}

\bibliographystyle{asa}    
\bibliography{ref_bib}   




\end{document}